\documentclass{article}
\usepackage{amsmath,amssymb,amsthm,enumitem}
\usepackage{graphicx}
\usepackage{hyperref}

\newtheorem{Theorem}{Theorem}[section]

\newtheorem{Corollary}[Theorem]{Corollary}
\newtheorem{claim}[Theorem]{Claim}
\newtheorem{Definition}[Theorem]{Definition}
\def\eps{\varepsilon}
\DeclareMathOperator\gen{gen}
\DeclareMathOperator\Prob{Prob}
\def\defeq{\buildrel{\mbox{\footnotesize def}}\over=}
\renewcommand\phi{\varphi}
\newcommand\I{\mathbb I}
\newenvironment{disp}{$$}{$$\ignorespaces}

\newenvironment{itemz}[1][3pt]{%
\setitemize{topsep=#1,noitemsep,leftmargin=\parindent,labelwidth=\parindent,labelsep=0pt,align=parleft}%
\begin{itemize}}{\end{itemize}}

\makeatletter
\def\correcthref{\hyper@anchor{\@currentHref}}
\makeatother
\def\contact#1#2#3#4{%
\gdef\makecontacts{\vskip 3mm \par \noindent #1 #2,\, #3\par e-mail: #4}}
\newcommand\keywords[1]{\par\vskip 2mm\noindent{\small{\sl Keywords: }
\begin{minipage}[t]{87mm}#1\end{minipage}}
\vskip 2mm}
\newcommand\classification[1]{\noindent{\small {\sl Classification: }
\begin{minipage}[t]{77mm}#1\end{minipage}}}

\begin{document}
\pagestyle{myheadings}
\title{Infinite Probabilistic Secret Sharing}
\author{Laszlo Csirmaz}
\date{}
\contact{Laszlo}{Csirmaz}{%
Alfred Renyi Institute of Mathematics, Budapest, and
Institute of Information Theory and Automation, Prague}
{csirmaz@renyi.hu}
\markboth{L.~Csirmaz}{L.~Csirmaz: Infinite Probabilistic Secret Sharing}
\maketitle

\renewcommand\v{\mathbf v}

\begin{abstract}
A probabilistic secret sharing scheme is a joint probability distribution of
the shares and the secret together with a collection of secret recovery
functions. The study of schemes using arbitrary probability spaces and
unbounded number of participants allows us to investigate their abstract
properties, to connect the topic to other branches of mathematics, and to
discover new design paradigms. A scheme is perfect if unqualified subsets
have no information on the secret, that is, their total share is independent
of the secret. By relaxing this security requirement, three other scheme
types are defined. Our first result is that every (infinite) access
structure can be realized by a perfect scheme where the recovery functions
are non-measurable. The construction is based on a paradoxical pair of
independent random variables which determine each other. Restricting the
recovery functions to be measurable ones, we give a complete
characterization of access structures realizable by each type of the
schemes. In addition, either a vector-space or a Hilbert-space based scheme
is constructed realizing the access structure. While the former one uses the
traditional uniform distributions, the latter one uses Gaussian
distributions, leading to a new design paradigm.
\classification{60B05, 94A62, 46C99, 54D10}
\end{abstract}


\section{Introduction}

The topic of this paper is secret sharing schemes where the domain of the
secret, the domain of the shares, or the set of players is not necessarily
finite. This type of approach, namely studying infinite objects instead of
finitary ones, is not novel even in the realm of cryptography, see, e.g.,
\cite{blakley-swanson, ChorKush, Makar:1980, patarin:infinite, Vaudenay:09}.
Further motivation and several examples can be found in
\cite{dibert-csirmaz:examples}. As can be expected, even finding the right
definition can be hard and far from trivial. We elaborate on this issue in
Section \ref{sec:conclusion}.

Secret sharing has several faces; it can be investigated equally from either
combinatorial or probabilistic point of view, see the survey paper
\cite{beimelsurvey}. The combinatorial view leads to \emph{set theoretical
generalizations} which are discussed in \cite{dibert-thesis}. In this paper
we take the probabilistic view and consider a secret sharing scheme as the
(joint) probability distribution of the shares and the secret. Defining
probability measures on arbitrary product spaces is not without problem, see
\cite{delarue, haezendonck, rokhlin} for a general description of the
problems, especially how and when the conditional distribution can be
defined. Our definitions avoid referring to conditional distributions at the
expense of a less transparent and less intuitive formulation. In Sections
\ref{sec:prob.secret.sharing} and \ref{sec:prod.spaces} we give all
necessary definitions from probability theory that will be used later on.
Nevertheless, a good working knowledge of measure theory and probability
spaces, as can be found, e.g., in \cite{kallenberg}, definitely helps.

A basic requirement in secret sharing -- usually called \emph{correctness}
-- is that qualified subsets, joining their shares, should be able to
recover the secret. The most straightforward way to ensure this property is
via \emph{recovery functions:} for each qualified subset $A$ there is a
function $h_A$ which, given the shares of members of $A$, returns the value
of the secret. In the classical case high complexity recovery functions can
only make the scheme more efficient. Quite surprisingly, this is not true in
general. In Section \ref{sec:non-measurable} we present a scheme in which
every share determines the secret, while, at the same time, every collection
of the shares is independent of the secret. This latter property is
interpreted as that the shares give ``no information'' on the secret, and
considered to be the strongest security requirement. Such a pathological
situation can be avoided by requiring the recovery functions to be
\emph{measurable}. This is exactly what we do: we focus on \emph{measurable}
schemes where all recovery functions are measurable.

Depending on how much information an unqualified subset might have on the
secret -- the \emph{security requirement} --, we define four scheme types. In
a \emph{perfect scheme} unqualified subsets should have no information at
all, meaning that the conditional distribution of the secret, given the
shares of the subset, is the same as the unconditional distribution. The
scheme is \emph{weakly perfect}, if, for some constant $c\ge 1$, the ratio of
the conditional and unconditional probabilities is always between $1/c$ and
$c$. A perfect scheme is a weakly perfect scheme with $c=1$. Weakly perfect
schemes were introduced in \cite{ChorKush} where they were called
``$c$-schemes.''

The scheme is \emph{ramp} when the constant which bounds the ratio of
conditional and unconditional probabilities is not necessarily uniform but
might depend on the unqualified set (but not on the value of the actual
shares). Finally the scheme is \emph{weakly ramp} if the constant $c$ might
depend on the actual values of the shares as well. The last case can be
rephrased as unqualified subsets cannot exclude any secret value with
positive probability.

In Sections \ref{sec:main} and \ref{sec:ramp} we characterize access
structures which can be realized by schemes of these types. We have both
topological and structural characterizations. Subsets of the set $P$ of the
participants can be considered as elements of the product $\{0,1\}^P$,
therefore an access structure -- the collection of qualified sets -- is a
subset of this space. Equipping $\{0,1\}^P$ with some topology we can speak
about the topological properties of the access structure. The \emph{Sierpi\'nski topology} \cite {Watson1990299} is especially promising. If $P$
is finite, a collection of subsets of $P$, as a subset of the topological
space $\{0,1\}^P$ is open if and only if the collection is upward closed,
which is a natural requirement for access structures. For definitions and
examples for this topology, see Section \ref{sec:sierpinski}. We prove that
a scheme can be realized by a perfect or weakly perfect scheme if and only
if it is an open set in this topology. Moreover, a scheme can be realized by
a ramp or a weakly ramp scheme if and only if it is $G_\delta$, that is, it
is the intersection of countably many open sets.

The structural characterization uses \emph{span programs} introduced in
\cite{spanprogram} and its generalization, \emph{Hilbert-space programs}. In
a span program we are given a vector space, a target vector, and every
participant is assigned one or more vectors. A structure realized by the
span program consists of those subsets of participants whose vectors span a
linear space containing the target vector. In a Hilbert-space program the
vector space is replaced by a Hilbert space, and a subset is qualified if the
target vector is in the closure of the linear span of their vectors. We
prove that exactly the open structures are realizable by span programs, and
exactly the $G_\delta$ structures are realizable by Hilbert-space programs.

Finally Section \ref{sec:conclusion} concludes the paper where we show that
not every access structure is realizable by a measurable scheme, discuss
additional scheme types, and list some open problems.


\section{Definitions}

This section defines access structures, then continues with Sierpi\'nski
topology and some basic properties of this topology. The definition of
probability secret sharing scheme is followed by properties of probability
measures on product spaces. Finally four scheme types are defined
corresponding to different security requirements. Motivations and examples
are omitted, they can be found, e.g., in \cite{dibert-csirmaz:examples}.


\subsection{Access structure}

An access structure $\mathcal A\subset 2^P$ is a non-trivial upward closed
(or monotone) family of subsets of the set $P$ of participants. To avoid
trivialities, an access structure does not contain singletons, and is not
empty. Given the collection $\mathcal A_0\subset 2^P$, the access structure
\emph{generated by $\mathcal A_0$} is
\begin{disp}
\gen(\mathcal A_0) \defeq  \{ A\subseteq P : \mbox{
$B\subseteq A$ for some $B\in\mathcal A_0$}\,\}.
\end{disp}
By monotonicity, an access structure is determined uniquely by any of its
generators. The access structure $\mathcal A$ is \emph{finitely generated} if
generated by a collection of finite subsets of $P$.


\subsection{Sierpi\'nski topology}\label{sec:sierpinski}

The \emph{Sierpi\'nski space} is a topological space defined on the two
element set $2=\{0,1\}$, where the open sets are the empty set, $\{1\}$, and
$\{0,1\}$. This topology is $T_0$, but not $T_1$, and is universal in the
sense that every $T_0$ space can be embedded into a high enough power, see
\cite{Watson1990299}. As only the Sierpi\'nski topology is used,
all topological notions in this paper refer to this topology.
The elements of the product topological space $\{0,1\}^P$ are the
characteristic functions of subsets of $P$, so its points can be identified
with the subsets of $P$. Consequently a collection $\mathcal A$ of subsets of
$P$ naturally corresponds to a subset of $\{0,1\}^P$. The following claim is
an easy consequence of the definition of the product topology.

\begin{claim}\label{claim:opentopology}
The collection $\mathcal A\subseteq 2^P$ is open in $\{0,1\}^P$ if and only if it
is a finitely generated monotone structure. \qed
\end{claim}
\noindent
In particular, if $P$ is finite, then a non-trivial $\mathcal A\subset 2^P$
is an access structure if and only if it is open.
\begin{Definition}\label{def:gdelta}\upshape
A set $\mathcal A\subseteq 2^P$ is $G_\delta$ if it is the intersection of
countably many open sets.
\end{Definition}

\begin{claim}\label{claim:gdelta}
$\mathcal A\subseteq 2^P$ is $G_\delta$ if and only if there are families
$\mathcal B_1\supseteq \mathcal B_2\supseteq\cdots$ consisting
of finite subsets of $P$ such that
\begin{disp}
    A\in\mathcal A ~~\Leftrightarrow~~ A\in\gen(\mathcal B_i)~~
            \mbox{for all $i$,}
\end{disp}
or, in other words, $\mathcal A = \bigcap_i \gen(\mathcal B_i)$.
\end{claim}
\begin{proof}
As $\mathcal B_i$ has only finite elements, $\gen(\mathcal B_i)$ is open,
and then $\bigcap_i \gen(\mathcal B_i)$ is $G_\delta$.

In the other direction, assume $\mathcal A = \bigcap U_i$, where $U_i$ is
open. Define $V_i = \bigcap \{ U_j: j\le i\}$, and put
\begin{disp}
    \mathcal B_i \defeq \{ A\subseteq P:\, A\in V_i,~~\mbox{$A$ is
finite}\,\}.
\end{disp}
As $V_i$ is open, $V_i=\gen(\mathcal B_i)$, and, of course, $\mathcal A =
\bigcap_i V_i$ as well. Moreover $V_{i+1}\subseteq V_i$, thus $V_i$ contains
every finite set that $V_{i+1}$ does.
\end{proof}

As an example, suppose $P$ is infinite and let $\mathcal A$ be
the family of all infinite subsets of $P$. Then $\mathcal A$ is not open,
but it is $G_\delta$ as it is the intersection of the families generated by
the $n$-element subsets of $P$ -- all of which are open.

For another example let $A_1$, $A_2$, $\dots$ be disjoint infinite subsets
of $P$ and let $\mathcal A$ be the family generated by these subsets. Then
$\mathcal A$ is upward closed, but it is not $G_\delta$. To show this,
suppose otherwise, and let $\mathcal B_1\supseteq \mathcal
B_2\supseteq\cdots$ be the families as in Claim \ref{claim:gdelta}. As
$A_i\in\mathcal A \subseteq \gen(\mathcal B_i)$, there is a (finite)
$B_i\in\mathcal B_i$ with $B_i\subseteq A_i$. Consider the set
$B=\bigcup_i B_i$. Clearly $B\in\gen(\mathcal B_i)$ as $B_i\in\mathcal B_i$
is a subset of $B$, thus $B\in\bigcap_i\gen(\mathcal B_i)$. On the other
hand, $B\cap A_i = B_i$ is finite thus $B$ does not extend any $A_i$, and
therefore it is not an element of $\mathcal A$.

In the third example we have countably many \emph{forbidden} subsets $F_1$,
$F_2$, $\dots$, and $\mathcal A$ consists of those subsets which are not
covered by any of the forbidden sets:
\begin{disp}
   \mathcal A = \{ A\subseteq P:\, A\not\subseteq F_i, ~~ i=1,2,\dots\,\}.
\end{disp}
$\mathcal A$ is obviously upward closed, and it is also $G_\delta$. To
conclude so, it is enough to show that 
\begin{disp}
    \mathcal A_n = \{A\subseteq P:\, A\not\subseteq F_i, ~~ i=1,2,\dots,n\,\}
\end{disp}
is open, as clearly $\mathcal A = \bigcap_i\mathcal A_i$. But $A\in\mathcal
A_n$ iff $A$ has a point outside $F_1$, a point outside $F_2$, \dots, a
point outside $F_n$. That is, $A$ has a subset with at most $n$ elements
which is also in $\mathcal A_n$. Therefore $\mathcal A_n$ is finitely
generated, that is, it is open.


\subsection{Probabilistic secret sharing scheme}\label{sec:prob.secret.sharing}

A \emph{secret sharing scheme} is a method to distribute some kind of
information among the participants $P$ so that qualified subsets could
recover the secret's value from their shares -- the scheme is
\emph{correct}, while forbidden subsets have no, or limited information on
the secret -- the \emph{security} requirement. In probabilistic schemes the
shares and the secret come from a (joint) probability distribution on the
product space of the corresponding domains.

\begin{Definition}\label{def:sss-domain}\upshape
The \emph{domain of secrets} is $X_s$, and the \emph{domain of shares} of
the participant $i\in P$ is $X_i$. We always assume that none of these sets
is empty, and $X_s$ has at least two elements, i.e. there is indeed a secret
to be shared.
\end{Definition}

To make our notation simpler, we denote $P\cup\{s\}$ by $I$ for the set of
\emph{indices}. If $A\subseteq P$ then $As$ denotes the set $A\cup\{s\}$, in
particular, $I=Ps$. We put $X = \prod_{i\in I} X_i$, and for a
subset $J\subseteq I$ we let $X^J=\prod_{i\in J}X_i$ be the restriction
of $X$ into coordinates in $J$.

\smallskip

Informally, a probabilistic secret sharing scheme is a probability
distribution on the set $X$ together with a collection of recovery
functions. Equivalently, it can be considered as a collection of random
variables $\{\xi_i: i\in I\}$ with some joint distribution so that $\xi_i$
takes values from $X_i$. The \emph{share} of participant $i\in P$ is the
value of $\xi_i$, and the \emph{secret} is the value of $\xi_s$.

\begin{Definition}\label{def:probabilistic-sss}\upshape
A \emph{probabilistic secret sharing scheme} is a pair $\mathcal S = \langle
\mu,h\rangle$, such that $\mu$ is a probability measure on the product space
$X = \prod_{i\in I} X_i$, where $I=P\cup\{s\}$, $X_s$ is the set of
(possible) secrets, and $X_i$ is the set of (possible) shares for
participant $i\in P$; and $h$ is the collection of \emph{recovery
functions}: for each $A\subseteq P$ the deterministic function $h_A:
X^A\to X_s$ gives a secret value given the shares of members of $A$.
\end{Definition}

When the dealer uses the scheme $\mathcal S = \langle\mu,h\rangle$, she
chooses an element $x\in X$ according to the given distribution $\mu$, sets
the secret to be $\xi_s = x(s)$, the $s$-coordinate of $x$, and send
privately the participant $i\in P$ the share $\xi_i = x(i)$, the $i$-th
coordinate of $x$. When members of $A\subseteq P$ want to recover the
secret, they use the recovery function $h_A$ on their shares to pinpoint the
secret value.

The scheme is \emph{correct} if qualified subsets recover the secret value,
at least with probability 1. To formalize this notion, we look at the
distribution of shares of $A$ and the secret value computed by the recovery
function $h_A$. When $x\in X$ is a distribution of all shares, the
\emph{projection $\pi_A(x)$} is its restriction to coordinates (indices) in
$A$, and the recovery function gives the secret value $h_A(\pi_A(x))\in
X_s$. So the probability of those sequences $x$ for which $x_s$ equals this
value must be 1.

\begin{Definition}\label{def:realizing-structure}\upshape
The scheme $\mathcal S=\langle\mu,h\rangle$ is \emph{correct} for the access
structure $\mathcal A\subset 2^P$ if for all $A\in\mathcal A$,
\begin{disp}
   \mu\big(\{x\in X:\, h_A(\pi_A(x)) = x_s \}\big) = 1.
\end{disp}
\end{Definition}

In a correct scheme the recovery functions of qualified sets are determined
almost uniquely. Indeed, let $h_A$ and $h^*_A$ be two correct recovery
functions. The set of those points where $h_A$ and $h^*_A$ differ is a
subset of
\begin{disp}
  \{ x\in X:\, h_A(\pi_A(x))\not=x_s\} \cup \{ x\in
X:\,h^*_A(\pi_A(x))\not=x_s \},
\end{disp}
and both sets have measure zero. It follows that the recovery functions form
a \emph{coherent family} in the following sense: if $A$ is qualified and
$A\subseteq B$, then $h^*_B(y) = h_A(\pi_A(y))$ is also a correct recovery
function, thus it must be equal to $h_B$ almost everywhere.

A secret sharing scheme must also provide \emph{security}, meaning that
unqualified subsets should have no or limited information on the secret. As
the precise definition requires some preparations from Probability Theory,
we postpone it to Section \ref{subsec:sss-types}.


\subsection{Probability measure on product spaces}\label{sec:prod.spaces}

As usual in probability theory \cite{kallenberg}, the definition of a
probability measure $\mu$ on the product space $X=\prod_i X_i$ requires a
$\sigma$-algebra $\Sigma$ on $X$. Let $J$ be a subset of $I$, then $X^J =
\prod_{i\in J} X_i$. A \emph{cylinder} is a set of the form $C=U\times
\prod_{i\notin J}X_i$ where $U \subseteq \prod_{i\in J}X_i$ is the
\emph{base} of the cylinder, and $J$ is its \emph{support}. Let moreover
\begin{disp}
  \Sigma^J = \big\{ E\subseteq X^J:
             (E\times \prod_{i\not\in J}X_i) \in \Sigma \big\} .
\end{disp}
It is easy to check that $\Sigma^J$ is a $\sigma$-algebra on $X^J$ if
$\Sigma$ is a $\sigma$-algebra on $X$. For each $J\subset I$ the
\emph{projection function $\pi_J$} maps the element $x\in X$ into $X^J$
keeping those coordinates of $x$ which are in $J$. With this notation a
subset $E$ of $X_J$ is in $\Sigma^J$ if and only if its inverse image under
$\pi_J$ is in $\Sigma$, namely, if $\pi^{-1}_J(E)\in\Sigma$.
The $\sigma$-algebra on the product space $X$ should be generated
by its finite-support cylinders, i.e. all sets from $\Sigma$ of the form
\begin{disp}
   U\times \prod_{i\notin J} X_i, ~~~ \mbox{where $J$ is finite and
$U\in\Sigma^J$}.
\end{disp}
 
Let $\mu$ be a probability measure on $\langle X,\Sigma\rangle$. Elements of
$\Sigma$ are the \emph{events}, and the probability of the event
$E\in\Sigma$ is just $\mu(E)$. As usual, $\mu$ is \emph{completed}, that is,
not only elements of $\Sigma$ have probability, but subsets of zero
probability events are also measurable. This means that for each
$\mu$-measurable $U\subseteq X$ there is a $V\in\Sigma$ such that the
symmetric difference of $U$ and $V$ is a $\mu$-zero set (i.e., it is a
subset of a set in $\Sigma$ with $\mu$-measure zero).

For a subset $J\subseteq I$ the \emph{marginal probability} is provided by
the probability measure $\mu^J$ defined on $X^J$ as follows. $E\subseteq
X^J$ is $\mu^J$-measurable iff $\pi_J^{-1}(E)$ is $\mu$-measurable, and
\begin{disp}
    \mu^J(E) = \mu\big(\pi_J^{-1}(E) \big). 
\end{disp} 
If $J$ has a single element $J=\{j\}$ then we also write $\mu_j$ instead of
$\mu^{\{j\}}$. In particular, $\mu_s$ is the marginal measure on the set of
secrets. With this notation, if $C$ is a cylinder with support $J$ and base
$U\in\Sigma^J$, then $\mu(C)=\mu^J(U)$.

As the probability measure $\mu$ determines the joint distribution of the
random variables $\xi_i$ for $i\in I$ (that is the $\sigma$-algebra $\Sigma$
on the whole space $X$ as well as the $\sigma$-algebras on each $X^J$)
uniquely, we can, and will, use this measure $\mu$ only in probabilistic
secret sharing schemes.

\smallskip

The following essential facts about probability measures will be 
used frequently and without further notice.
\begin{claim}\correcthref\label{claim:basic-facts}
\begin{itemz}
\item[\upshape(a)] For each $E\in\Sigma$ there is a countable set
$J\subseteq I$ such that $E=\pi_J^{-1}\big(\pi_J(E)\big)$, that is, $E$ is a
cylinder with countable support.
\item[\upshape(b)] 
For any $\mu$-measurable set $E\subseteq X$ and any $J\subseteq I$,
$\mu^J\big(\pi_J(E)\big) \ge \mu(E)$.
\item[\upshape(c)]
 For any $\mu$-measurable $E\subseteq X$ and $\eps>0$ there is a cylinder
$E'$ with finite support such that $\mu(E-E')=0$, and $\mu(E'-E)<\eps$.
\item[\upshape(d)]
 For any $\mu$-measurable $E\subseteq X$ and $\eps>0$ there is a finite
$J\subseteq I$ such that $\mu(E)\le \mu^J(\pi_J(E))<\mu(E)+\eps$.
\end{itemz}\end{claim}
\begin{proof}
(a) Cylinders with finite support have the stated property. Also, this
property is preserved by taking complements and countable unions. Thus all
elements in the smallest $\sigma$-algebra generated by finite support 
cylinders have the property claimed.

(b) The statement is immediate from the fact that
$\pi^{-1}_J\big(\pi_J(E)\big) \supseteq E$.

(c) By part (a), any $\mu$-measurable $E\subseteq X$ is, up to a set of
measure zero, a cylinder $C$ with countable support. Thus it is the
intersection of the finite support cylinders $C_n =
\pi_{J_n}^{-1}\big(\pi_{J_n}(C)\big)$ where $J_n$ is the set of first $n$
elements of the support of $C$. As $C_{n+1}\subseteq C_n$,
$\lim_{n\to\infty} \mu(C_n) = \mu(C)$ decreasingly, and the claim follows.

(d) The first inequality comes from (b). By (c), there is a cylinder $E'$
with finite support such that $E-E'$ is a zero set, while
$\mu(E')<\mu(E)+\eps/2$. As $\mu(E-E')=0$, there is a zero set $Z\in\Sigma$
such that $Z\supseteq E-E'$. By (a) $Z$ is a cylinder with countable
support, thus there is a finite support cylinder $E''\supseteq Z$ with
$\mu(E'')<\eps/2$. Let $J$ be the (finite) support of $E'\cup E''$, then
$\mu^J(\pi_J(E'))=\mu(E')$ and $\mu^J(\pi_J(E''))=\mu(E'')$. As $E\subseteq
E'\cup Z\subseteq E'\cup E''$,
\begin{eqnarray*}
  \mu^J(\pi_J(E))&\le& \mu^J(\pi_J(E'\cup E'')) \\
       &\le& \mu^J(\pi_J(E')) + \mu^J(\pi_J(E''))
       = \mu(E')+\mu(E'') \\
      &<& (\mu(E)+\eps/2) + \eps/2 = \mu(E)+\eps, 
\end{eqnarray*}
as was required.
\end{proof}

Let $B\subseteq P$ be any subset of participants. The collective set of
shares they receive falls into the (measurable) set $U\subseteq X^B$ with
probability $\mu^B(U)$. Similarly, if $E\subseteq X_s$ is measurable, then
the probability that the secret falls into $E$ is $\mu_s(E)$. The
\emph{conditional probability distribution} of the secret, assuming that the
shares of $B$ come from the set $U$ with $\mu^B(U)>0$, is defined as
\begin{disp}
    \mu_s(E | U) = \frac{\mu^{Bs}(U\times E)}{\mu^B(U)}.
\end{disp}
Here we wrote $Bs$ for $B\cup\{s\}$. Observe that
$\mu^s(E|X^B) = \mu_s(E)$, and $\mu_s(\cdot | U)$ is a probability measure
on $X_s$.

It would be tempting to define the conditional distribution given not a
(measurable) subset of the shares, but the shares themselves. Unfortunately
such conditional distributions do not always exist \cite{ChangPollard97},
nevertheless in statistics their existence is almost always assumed.
Fortunately, at the expense of a slightly more complicated and less
intuitive formulation, we can avoid those conditional distributions.


\subsection{Security requirements}\label{subsec:sss-types}

In a secret sharing scheme unqualified subsets are required to have no, or
limited information on the secret. Depending on how strong the security
guarantee is we distinguish four scheme types.

\begin{Definition}\label{def:scheme-type}\upshape
Let $\mathcal S = \langle\mu,h\rangle$ be a secret sharing scheme on the 
set $P$ of participants. The
scheme is \emph{perfect}, \emph{weakly perfect}, \emph{ramp}, or
\emph{weakly ramp} if the collective set of shares of an unqualified subset 
$B\subseteq P$ satisfies the following condition:
\begin{itemz}
\item[\hbox{\bf perfect.\space}]
$B$ gets no information on the secret, meaning that the set of shares and
the secret are (statistically) independent. That is, for every measurable
$U\subseteq X^B$ and $E\subseteq X_s$ we have
\begin{disp}
    \mu^{Bs}(U\times E) = \mu^B(U)\cdot\mu_s(E).
\end{disp}
This can also be expressed as the conditional probability $\mu_s({\cdot}|U)$
coincides with the unconditional probability $\mu_s({\cdot})$ for all
$U\subseteq X^B$ with $\mu^B(U)>0$.

\item[\hbox{\bf weakly perfect.\space}] 
For $U\subseteq X^B$ the conditional probability $\mu_s({\cdot}|U)$ deviates
from $\mu_s(\cdot)$ by a constant factor only, i.e., for some positive
constant $c\ge1$ (independently of the unqualified set $B$), for all
measurable $U \subseteq X^B$ and $E \subseteq X_s$,
\begin{equation}\label{eq:SchemaType:1}
    \frac1c\cdot\mu^{B}(U)\cdot\mu_s(E) \le \mu^{Bs}(U\times E) 
  \le c\cdot\mu^{B}(U)\cdot\mu_s(E) .
\end{equation}

\item[\hbox{\bf ramp.\space}]
The constant $c=c_B$ in {\upshape(\ref{eq:SchemaType:1})} might depend on
the subset $B$ (but not on $U$ and $E$).

\item[\hbox{\bf weakly ramp.\space}]
Based on their collective shares, $B$ cannot exclude any subset of the 
secrets with positive measure:
\begin{disp}
    \mu^B(U)\cdot\mu_s(E) > 0 ~~\mbox{implies}~~ \mu^{Bs}(U\times E)>0.
\end{disp}
(Observe that the reverse implication always holds.)
\end{itemz}
\end{Definition}

These definitions reflect and extend the usual ones in classical secret
sharing schemes. The traditional requirement for perfect schemes is the
statistical independence as defined here. Weakly perfect schemes were
introduced in \cite{ChorKush}, where such schemes with constant
$c$ are called ``$c$-schemes.'' No universally accepted definition exists
for ramp schemes. The best approach is that in a ramp scheme under no
circumstances an unqualified subset should be able to recover the secret.
Our definitions reflect this idea. However, see the discussion in Section
\ref{sec:conclusion}.

When the scheme $\mathcal S$ is classical, namely the number of participants
is finite and both the shares and the secret come from a finite domain
(that is, $X$ is finite), then the conditions for weakly perfect, ramp, and
weakly ramp schemes are equivalent, while not equivalent to perfect schemes.

\begin{claim}\label{claim:forbidden}
The types above are listed in decreasing strength, namely
\begin{center}
  perfect ${}\Rightarrow{}$ weakly perfect 
    ${}\Rightarrow{}$ ramp
    ${}\Rightarrow{}$ weakly ramp.
\end{center}
None of the implications can be reversed.
\end{claim}
\begin{proof}
It is not difficult to construct schemes witnessing the irreversibility of
these implications. For concrete examples consult
\cite{dibert-csirmaz:examples}.
\end{proof}


\section{Non-measurable schemes realize all}\label{sec:non-measurable}

The probabilistic secret sharing scheme $\mathcal S=\langle \mu,h\rangle$ is
\emph{measurable} if all recovery functions $h_A$ are measurable. Requesting
measurability seems to be a technical issue. It is not, as is shown by
Theorem \ref{thm:non-measurable}. The proof uses a paradoxical construction
of two random variables due to G\'abor Tardos, and is included here with his
permission.

\begin{Theorem}[G.~Tardos]\label{thm:2-paradoxical}
Let $\I$ denote the unit interval $[0,1]$.
There are two random variables $\xi$ and $\eta$ with at joint distribution
on $\I\times \I$ such that
\begin{itemz}
\item[\upshape (a)] both $\xi$ and $\eta$ are uniformly distributed on $\I$,
\item[\upshape (b)] $\xi$ and $\eta$ are independent,
\item[\upshape (c)] both of them determine the other's value.
\end{itemz}
\end{Theorem}
\begin{proof}
The idea of the construction is to find a subset $H\subseteq \I\times \I$
with the following properties:
\begin{itemz}
\item[(i)] $H$ is a graph of a bijection from $\I$ to $\I$,
\item[(ii)] $H$ has a point in every positive (Lebesgue) measurable subset of
$\I\times \I$.
\end{itemz}
When we have such an $H$, then define the $\sigma$-algebra $\Sigma$ on $H$
as the trace of the (Lebesgue) measurable sets of $\I\times \I$, and define the
probability measure $\mu$ on $H$ as
\begin{disp}
    \mu(U\cap H) = \lambda(U) 
\end{disp}
whenever $U$ is a measurable subset of $\I\times \I$.
This definition is sound as if $U_1\cap H=U_2\cap H$ for two measurable
subsets $U_1$ and $U_2$,
then property (ii) ensures $\lambda(U_1)=\lambda(U_2)$. Let
$(\xi,\eta)$ be a random element of $H$ distributed according to the measure 
$\mu$. As $H$ is a graph of a bijective function, property (c) holds.
Now let $E\subseteq \I$ be (Lebesgue) measurable. Then
\begin{disp}
    \Prob(\xi\in E) = \mu( H\cap (E\times \I) ) = \lambda(E\times
\I)=\lambda(E),
\end{disp}
thus $\xi$ is indeed uniformly distributed on $\I$, and similarly for $\eta$.
Finally, let $E$ and $F$ be measurable subsets of $\I$. Then
\begin{eqnarray*}
  \Prob(\,\xi\in E \mbox{~and~} \eta\in F) &=&
    \mu( H\cap ( E\times F) ) \\
   &=& \lambda(E\times F) = \lambda(E)\cdot \lambda(F) \\
   &=& \Prob(\,\xi\in E) \cdot \Prob(\,\eta\in F),
\end{eqnarray*}
which shows that $\xi$ and $\eta$ are independent.

Thus we need to find a subset $H\subset \I\times \I$ satisfying (i) and
(ii). We will use transfinite induction (thus the axiom of choice) to add
points of $H$. First note that every positive measurable set contains a
positive closed set, and there are only continuum many closed sets. Let
$F\subseteq \I\times \I$ be closed and positive, then $F$ contains a
generalized continuum by continuum grid. Namely, there are subsets $U$,
$V\subseteq \I$ such that both $U$ and $V$ have continuum many elements and
$U\times V\subseteq F$. Using these properties we proceed as follows.

Enumerate all closed positive sets as $F_\alpha$, and all real numbers in
$\I$ as $x_\alpha$ where $\alpha$ runs over all ordinals less than
continuum. At each stage we add at most three new points to $H$. Suppose we
are at stage indexed by $\alpha$. As there is a continuum by continuum grid
in $F_\alpha$ and until so far we added less than continuum many points to
$H$, there is a point in $F_\alpha$ such that neither its $x$ nor its
$y$-coordinate has been chosen as an $x$ (or $y$ respectively) coordinate of
any previous point. Add this element of $F_\alpha$ to $H$. Then look at the
real number $x_\alpha$. If there is no point in $H$ so far with an
$x$-coordinate (or $y$-coordinate) equal to $x_\alpha$, then add the point
$(x_\alpha,z)$ (the point $(z,x_\alpha)$) to $H$, where $z$ is not among the
$y$-coordinates ($x$-coordinates) of points in $H$ so far.

The set $H$ we constructed during this process satisfies properties (i) and
(ii). Indeed, every real number in $\I$ is a first (second) coordinate of
some element of $H$. During the construction we made sure that every
horizontal (vertical) line intersects $H$ in at most a single point. Thus
$H$ is indeed a graph of a bijection of $\I$. Finally $H$ contains a point
from each positive closed subset of $\I\times \I$, and thus from each
positive measurable subset as well.
\end{proof}

Remark that the bijection encoded by $H$ is not measurable in the product
space (which, incidentally, is the standard Lebesgue measure on $\I\times
\I$).

\begin{Theorem}\label{thm:non-measurable}
Given any access structure $\mathcal A\subset 2^P$, there is a perfect
(non-measurable) secret sharing scheme realizing $\mathcal A$.
\end{Theorem}
\begin{proof}
Take the pair of random variables $\langle\xi,\eta\rangle$ from Theorem
\ref{thm:2-paradoxical}. Give every participant $\xi$ as a share, and set
$\eta$ as the secret. Now $\xi$ determines $\eta$, therefore qualified
subsets can recover the secret.
Similarly, $\xi$ and $\eta$ are independent, therefore unqualified subsets
have ``no information on the secret.'' Consequently this is a perfect
probabilistic secret sharing scheme realizing $\mathcal A$. Note that it is
not measurable as the recovery function is not measurable.
\end{proof}


\section{Structures realized by perfect and weakly perfect schemes}\label{sec:main}

From this point on only measurable schemes are considered. This section
gives a complete characterization of access structures which can be realized
by perfect or weakly perfect measurable schemes as defined in Definition
\ref{def:scheme-type}. Recall that an access structure $\mathcal A\subset
2^P$ is \emph{open} if the qualified sets form an open set in the
Sierpi\'nski topology.

Monotone span programs were introduced by Karchmer and Wigderson
\cite{spanprogram}, and they are used to study linear schemes. To fit into
our framework we extend it by allowing infinitely many participants and
arbitrary vector spaces. Given a vector space $V$ and a subset $H\subset V$,
the \emph{linear span} of $H$ is the set of all (finite) linear combinations
of elements of $H$. The linear span is a linear subspace of $V$.

\begin{Definition}\label{def:spanprogram}\upshape
Let $P$ be the (possibly infinite) set of participants. A \emph{span
program} consists of a vector space $V$, a target vector $\v\in V$, and a
function $\phi:P\to 2^V$ which assigns a (not necessarily finite) collection
of vectors to each participant. The structure $\mathcal A\subset 2^P$ is
\emph{realized} by the span program if
\begin{disp}
    A\in\mathcal A ~~\Leftrightarrow~~
    \v\in \mbox{linear span of }\bigcup\{\,\phi(p): p \in A\} .
\end{disp}
\end{Definition}
\noindent
It is clear that structures realized by span programs are monotone and
finitely generated.

\begin{Theorem}\label{thm:perfect}
The following statements are equivalent for any access structure $\mathcal
A\subset 2^P$.
\begin{itemz}
\item[\upshape 1.]
  $\mathcal A$ is realized by a span program.
\item[\upshape 2.]
  $\mathcal A$ is realized by a perfect measurable probabilistic scheme.
\item[\upshape 3.]
 $\mathcal A$ is realized by a weakly perfect measurable probabilistic scheme.
\item[\upshape 4.]
 $\mathcal A$ is open.
\item[\upshape 5.]
 $\mathcal A$ is finitely generated.
\end{itemz}
\end{Theorem}
\begin{proof}
The equivalence $4\Leftrightarrow5$ is the statement of Claim 
\ref{claim:opentopology}. The implication $2\Rightarrow3$ is trivial, thus
we need to prove the implications $3\Rightarrow 5$, $5\Rightarrow 1$, and 
$1\Rightarrow 2$. 

\medskip\noindent
$3\Rightarrow5$:
We remark that $\mathcal A$ is finitely generated if and only if every
qualified set contains a finite qualified set. Suppose that the weakly
perfect measurable scheme $\mathcal S=\langle\mu,h\rangle$ realizes
$\mathcal A$ and let $c\ge 1$ be the constant from Definition
\ref{def:scheme-type}, equation (\ref{eq:SchemaType:1}).

Choose a subset $E_1\subset X_s$ of the secrets so that both $E_1$ and its
complement $E_2=X_s-E_1$ is positive:
\begin{disp}
    p_1=\mu_s(E_1)>0, ~~~~ p_2=\mu_s(E_2)>0,
\end{disp}
and, of course, $p_1+p_2=1$. Let $A\in\mathcal A$ be infinite, we must show
that it has a finite qualified subset. The recovery function $h_A$ is
measurable, thus the sets $U_i = h_A^{-1}(E_i)$ are measurable, and
$\mu^{As}(U_1\times E_2) = \mu^{As}(U_2\times E_1)=0$ as $h_A$ gives the
right secret with probability 1. Consequently
\begin{eqnarray*}
  \mu^A(U_1) &=& \mu^{As}(U_1\times X_s) =
     \mu^{As}(U_1\times E_1) + \mu^{As}(U_1\times E_2) ={}\\
  &=& \mu^{As}(U_1\times E_1) = {}\\
  &=& \mu^{As}(U_1\times E_1) + \mu^{As}(U_2\times E_1) ={}\\
  &=& \mu^{As}(X^A\times E_1) = \mu_s(E_1) = p_1.
\end{eqnarray*}
By item (d) of Claim \ref{claim:basic-facts}, for every positive $\eps>0$
there is a finite subset $B\subset A$ such that setting
$V_1=\pi_B(U_1)\subseteq X^B$,
\begin{disp}
      \mu^{Bs}(V_1\times E_2) < \mu^{As}(U_1\times E_2)+\eps = \eps,
\end{disp}
and, by item (b) of the same Claim,
\begin{disp}
   \mu^B(V_1) \ge \mu^A(U_1) = p_1 .
\end{disp}
Now we claim that if $\eps$ is small enough, then $B$ is qualified.
Indeed, 
$\mathcal S$ is weakly perfect with constant $c$, thus if $B$ were unqualified
then applying condition (\ref{eq:SchemaType:1}) for $V_1\subseteq X^B$ and 
$E_2\subseteq X_s$ we get
\begin{disp}
   \frac1c\cdot p_1\cdot p_2 \le \frac1c\cdot \mu^B(V_1)\cdot \mu_s(E_2)
    \le \mu^{Bs}(V_1\times E_2) < \eps.
\end{disp}
But this inequality clearly does not hold when $\eps$ is small enough,
proving the implication.

\medskip\noindent
$5\Rightarrow1$:
Suppose $\mathcal A\subset 2^P$ is finitely generated, say $\mathcal A=
\gen(\mathcal B)$, where every $B\in\mathcal B$ is finite. Let $V$ be a
large enough (infinite dimensional) vector space, and fix the target vector
$\mathbf v\in V$. We want to assign vectors to participants so that $\mathbf
v$ is in the linear span of the vectors assigned to members of $A\subseteq
P$ if and only if $A$ is qualified. This can be done as follows. For each
$B\in \mathcal B$ ($B$ is finite!) choose $|B|-1$ vectors from $V$ which are
linearly independent from everything chosen so far (including the target
vector), and set the $|B|$-th vector so that the sum of these $|B|$ many
vectors equals $\mathbf v$. Assign these vectors to the corresponding
members of $B$. A participant $p\in P$ will receive all vectors assigned to
him.

\medskip\noindent
$1\Rightarrow2$: 
If $\mathcal A\subset 2^P$ is realized by a span program, then it is
finitely generated. The proof of the implication $5\Rightarrow1$ above gives
the stronger result that if $\mathcal A$ is finitely generated, then it can
be realized by a span program in which the vector space $V$ is over some (in
fact, any) finite field. Thus if $\mathcal A$ can be realized by any span
program, then it can be realized by a span program over a finite field
$\mathbb F$.

Fix a base $H$ of the vector space $V$, and for each $\mathbf h\in H$ the
dealer picks $r_{\mathbf h}\in\mathbb F$ uniformly and independently (this
is where we need $\mathbb F$ to be finite). Write the goal vector in base
$H$ as the finite sum $\v=\sum_j \beta_j\mathbf h_j$ ($\mathbf h_j\in H$),
and set the secret to be $s=\sum_j \beta_jr_{\mathbf h_j}$.

Next, suppose the vector $\mathbf x$ was assigned (among others) to the
participant $p$. 
Write $\mathbf x$ as a (finite) linear combination of base elements:
$
    \mathbf x = \sum\nolimits_i \alpha_i \mathbf h_i,
$
and then the dealer gives $p$ the share the pair
$\langle \mathbf x, \sum_i \alpha_i r_{\mathbf h_i}
\rangle$. Thus $p$ receives a share (an element of $\mathbb F$ labeled by
the public vector $\mathbf x$) for each vector assigned to him.

It is clear that subsets of participants which have $\v$ in their linear
span can compute the secret (as an appropriate linear combination of some of
the shares), and shares of an unqualified set is independent of the secret,
as was required.
\end{proof}


\section{Structures realized by ramp and weakly ramp schemes}\label{sec:ramp}

This section characterizes access structures which can be realized by
measurable (weakly) ramp schemes. The characterization uses the notion of
Hilbert-space programs, which is similar to that of span programs, only the
vector space is replaced by a Hilbert space, and the target vector should be
in the \emph{closure} of the linear span rather than in the linear span of
the generating vectors.

We also prove a generalization of the main result of Chor and Kushilevitz
\cite{ChorKush} saying that if the scheme distributes infinitely many
secrets, then the share domain of important participants should be large.
Finally we give a ramp scheme which distributes infinitely many secrets,
while every share domain is finite. Of course, in this scheme no participant
can be important.

\begin{Definition}\label{def:hilbert-program}\upshape
A \emph{Hilbert-space program} consists of a Hilbert space $H$, a target
vector $v\in H$, and a function $\phi:P\to2^H$ which assigns a
subset of the Hilbert space to each participant. The structure $\mathcal
A\subset 2^P$ realized by the Hilbert-space program is
\begin{disp}
A\in\mathcal A ~~\Leftrightarrow~~
v\in \mbox{closure of the linear span of }\bigcup\{\phi(p):p\in A\} .
\end{disp}
\end{Definition}

\begin{Theorem}\label{thm:ramp}
The following statements are equivalent for any access structure $\mathcal
A \subseteq 2^P$.
\begin{itemz}
\item[\upshape 1.]
 $\mathcal A$ is realized by a Hilbert-space program;
\item[\upshape 2.]
 $\mathcal A$ is realized by a ramp measurable probabilistic secret
 sharing scheme;
\item[\upshape 3.]
 $\mathcal A$ is realized by a weakly ramp measurable scheme;
\item[\upshape 4.]
 $\mathcal A$ is $G_\delta$.
\end{itemz}
\end{Theorem}
\begin{proof}
The implication $2\Rightarrow3$ is trivial; we will show $1\Rightarrow2$,
$3\Rightarrow4$ and $4\Rightarrow1$. Also, we will use Claim
\ref{claim:gdelta} which gives an equivalent characterization of $G_\delta$
structures.

\medskip\noindent
$3\Rightarrow4$:
Let $\mathcal S=\langle\mu,h\rangle$ be a weakly ramp scheme which realizes
$\mathcal A\subset 2^P$. As in the proof of Theorem \ref{thm:perfect},
choose $E_1\subset X_s$, $E_2=X_s-E_1$ so that
\begin{disp}
    p_1=\mu_s(E_1)>0, ~~~~ p_2=\mu_s(E_2)>0, ~~~~ p_1+p_2=1.
\end{disp}
As the set of all participants is always qualified, and $h_P$ is measurable,
the sets $U_i=h_P^{-1}(E_i)\subseteq X_P$ are measurable,
$\mu^{Ps}(U_1\times E_2)=\allowbreak\mu^{Ps}(U_2\times E_1)=0$, and
\begin{disp}
  \mu^P(U_1)=\mu^{Ps}(U_1\times E_1)=p_1.
\end{disp}
Let us define the family $\mathcal B_n$ of finite subsets of $P$ as follows:
\begin{disp}
   B\in\mathcal B_n ~~\Leftrightarrow~~ 
     B \mbox{ is finite, and } \mu^{Bs}(\pi_B(U_1)\times E_2) < \frac1n .
\end{disp}
It is clear that $\mathcal B_{n+1}\subseteq \mathcal B_n$, thus $\mathcal B
= \bigcap_n\gen(\mathcal B_n)$ is $G_\delta$. We claim that a subset of
participants is qualified if and only it is in $\mathcal B$. First, let
$A\subseteq P$ be qualified. Then $g_A = h_P\circ\pi_A$ is a (measurable)
recovery function for $A$, thus letting $V_1=g_A^{-1}(E_1)$,
$\mu^{As}(V_1\times E_2)=0$, and then for each $n$ there is a finite
$B_n\subseteq A$ such that
\begin{disp}
    \mu^{B_ns}(\pi_{B_n}(V_1)\times E_2) < \frac1n .
\end{disp}
Observing that $V_1=\pi_A(U_1)$, we get that $A\in\gen(\mathcal B_n)$ for
each $n$, as was required. In the other direction, let $B\subseteq P$ be not
qualified, and let $V_1=\pi_B(U_1)\subseteq X^B$. As
$\mu^B(V_1)\ge\mu^P(U_1)=p_1>0$ and $\mu_s(E_2)=p_2>0$, the weakly ramp
property gives
\begin{disp}
    \mu^{Bs}(V_1\times E_2) = \mu^{Bs}(\pi_B(U_1)\times E_2) > 0.
\end{disp}
For any subset $B'$ of $B$, $\mu^{B's}(\pi_{B'}(U_1)\times E_2) \ge
\mu^{Bs}(V_1\times E_2)$, consequently $B$ is not in $\gen\mathcal B_n$ when
$n\ge 1/\mu^{Bs}(V_1\times E_2)$.

\medskip\noindent
$4\Rightarrow1$:
Let $\mathcal B_1\supseteq \mathcal B_2\supseteq\cdots$ be families of
finite subsets of $P$ such that $\mathcal A = \bigcap_n\gen(\mathcal B_n)$,
as given by Claim \ref{claim:gdelta}. Then $A\in\mathcal A$ if and only if
$A$ is in $\gen(\mathcal B_n)$ for infinitely many $n$. Let $H$ be a huge
dimensional (not separable) Hilbert space, and fix an orthonormal base
$e_1$, $e_2$, $\dots$, (countably many elements) plus $\{ \bar e_\alpha:
\alpha\in I\}$ for some index set $I$. The target vector will be
\begin{disp}
    v = e_1 + \frac{e_2}2 + \frac{e_3}3 + \cdots ,
\end{disp}
and let $v_n= \sum_{i=1}^n e_i/i$. For each (finite) $B\in \mathcal B_n$,
the first $|B|-1$ members of $B$ will be assigned new base elements from
among $\bar e_\alpha$, and the last member will be assigned an element from
$H$ so that the sum of these $|B|$ elements be equal to $v_n$.

The target vector is in the closure of the linear span of Hilbert space
elements assigned to members of $A\subseteq P$ if and only if $v_n$ is in
their linear span for infinitely many $n$. But this latter event happens if
and only if $A$ is in $\gen(\mathcal B_n)$, thus this Hilbert-space program
realizes $\mathcal A$, as required.

\medskip\noindent
$1\Rightarrow2$:
Let $H$ be the (real) Hilbert space over which the program is defined, and
fix an orthonormal base $\{e_\alpha :\,\alpha\in I\}$ of $H$. For each
element in this base assign a standard normal random variable $\xi_\alpha$
so that they are totally independent. An element $a\in H$ can be written as
\begin{disp}
   a = \sum \lambda_\alpha e_\alpha, ~~~\mbox{where } \sum \lambda^2_\alpha <
                \infty.
\end{disp}
Assign the (random) variable $\xi_a =\sum \lambda_\alpha\xi_\alpha$ to this
element $a\in H$. More information about these \emph{Gaussian spaces} can be
found in \cite{janson1997gaussian}. We list here only some basic properties
which will be needed for our construction.

The random variable $\xi_a$ is normal with expected value 0 and variance
$\|a\|^2$, furthermore $\xi_a$ and $\xi_b$ are independent if and only if
$a$ and $b$ are orthogonal. If $v$ is in the closure of the linear span of
$E\subseteq H$, then $\xi_v$ is determined (with probability 1) by the
values of $\{ \xi_a :\, a\in E\}$.

Let $L\subseteq H$ be a closed linear subspace. Any $v\in H$ has an
orthogonal decomposition $v=v_1+v_2$ such that $v_1 \perp L$ and $v_2\in L$.
If $v_1\not=0$ then $\xi_v$ has a conditional distribution given the values
of all $\xi_a$ for $a\in L$, and this distribution is normal with variance
$\|v_1\|^2$ (the expected value depends on the values of the variables
$\xi_a$).

We define a secret sharing scheme $\mathcal S$ realizing $\mathcal A$ as
follows. Every domain will be either the set of reals or some power of the
reals. Let $v\in H$ be the target vector. The secret is the value of
$\xi_v$. The share of participant $p\in P$ is the collection of the values
of $\xi_a$ for all elements $a\in H$ assigned to $p$.

If $A\subseteq P$ is qualified, then $v$ is in the closure of the linear
span, thus $\xi_v$ is determined by the shares of $A$. If $B\subseteq P$ is
unqualified, then the target vector is not in the closure of the linear
span, let $v_1\not=0$ be its orthogonal component. The conditional
distribution of the secret, \emph{given all shares of $B$}, is normal with
$\|v_1\|^2$ variance. As the density function of the normal distribution is
nowhere zero, the probability that the secret is in the set $E\subseteq R$,
both in the unconditional and in the conditional case, is zero if and only
if $E$ is a zero set. Consequently this $\mathcal S$ is a \emph{weakly ramp
scheme} realizing $\mathcal A$. It is easy to see that this scheme is never
ramp as the ratio of the conditional and unconditional distribution function
is never bounded.

However, one can twist this scheme to be a ramp one. The only change is to
set the secret to be the \emph{fraction part} of $\xi_v$, see this trick in
\cite{dibert-csirmaz:examples}. As the density function of the fractional
part of a normal distribution is bounded (there is a $c\ge1$ such that it is
between $1/c$ and $c$) and the bound depends on the variation only, the
conditional distribution of the secret, given the shares of an unqualified
set, is bounded, where the bound depends on the subset only, and not on the
actual values of the shares. Consequently this scheme is a ramp scheme
realizing $\mathcal A$.
\end{proof}

Next we prove a generalization of the main result of Chor and Kushilevitz 
\cite{ChorKush}. It follows from a slightly more general statement which we
prove first.

\begin{Theorem}\label{thm:atomless}
Suppose $\mathcal S$ is a measurable ramp scheme, $A$ and $B$ are disjoint
unqualified sets such that $A\cup B$ is qualified. Suppose moreover that
there are infinitely many secrets. Then $\mu^A$ is atomless.
\end{Theorem}

An immediate consequence is that under the same conditions the set of shares
of $A$, namely $X^A$, must have cardinality (at least) continuum.

\begin{proof}
Suppose by contradiction that $X^A$ is atomic and $\mu^X(\{a\})>0$ for some
$a\in X^A$. 
Partition the set of secrets into countably many positive sets as
$X_s=\bigcup_i E_i$ where $\mu_s(E_i)$ is positive. Let $h:\,\allowbreak X^A\times X^B
\mapsto X_s$ be the function which
determines the secret given the shares of $A$ and $B$. Let 
\begin{disp}
  V_i = \{ y \in X^B: h(a,y)\in E_i \}.
\end{disp}
As $h$ is measurable, each $V_i$ is measurable, moreover the sets $\{a\}
\times V_i\times E_i$ and $\{a\}\times X^B\times E_i$ have the same measure.
Using the boundedness property for $A$ we get
\begin{eqnarray*}
    \mu^{Bs}(V_i\times E_i) &\ge& \mu^{ABs}(\{a\}\times V_i\times E_i) \\
        &=& \mu^{ABs}(\{a\}\times X^B \times E_i ) \\
        &=& \mu^{As}(\{a\}\times E_i) \\[2pt]
        &\ge& \frac1{c_A} \cdot \mu^A(\{a\}) \cdot\mu_s(E_i).
\end{eqnarray*}
Applying the boundedness twice for $B$ we have
\begin{eqnarray*}
  \mu^{Bs}(V_i\times E_1) &\ge& \frac1{c_B}\cdot \mu^B(V_i)\cdot \mu_s(E_1) \\
        &=& \frac{\mu_s(E_1)}{c^2_B\cdot\mu_s(E_i)} \cdot
           c_B\mu^B(V_i)\cdot\mu_s(E_i) \\
        &\ge& \frac{\mu_s(E_1)}{c^2_B\cdot\mu_s(E_i)}\, \mu^{Bs}(V_i\times
            E_i) \\
        &\ge& \frac1{c^2_B c_A} \cdot \mu_s(E_1)\cdot \mu^A(\{a\}),
\end{eqnarray*}
where we used $\mu(E_i)>0$ and the previous estimate in the last step.
As $h$ is defined on $X^A\times X^B$ and $\bigcup_i E_i = X_s$, we have
$\bigcup_i V_i = X^B$, furthermore the $V_i$'s are pairwise disjoint. Thus
\begin{disp}
   1\ge \mu^{Bs}(X_B\times E_1) = \sum_i \mu^{Bs}(V_i\times E_1)
     \ge \sum_i \; \big( \frac1{c^2_Bc_A} \cdot \mu_s(E_1)\cdot \mu^A(\{a\}) \big),
\end{disp}
which can happen only when $\mu^A(\{a\})=0$, a contradiction.
\end{proof}

A participant $p\in P$ is \emph{important} if
there is an unqualified set $B\subseteq P$ such that $B\cup\{p\}$ is
qualified.

\begin{Corollary}\label{corr:gen-chorkush}
Suppose $\mathcal S$ is a measurable ramp scheme which distributes infinitely
many secrets. Then the share domain of every important participant must have
cardinality at least continuum.
\end{Corollary}
\begin{proof}
By assumption, no singleton is qualified, thus we can apply Theorem
\ref{thm:atomless} with $A=\{p\}$ and the unqualified $B$ such that $A\cup
B$ is qualified. As $\mu_p$ is atomless, $X_p$ must have at least continuum
many elements.
\end{proof}

Surprisingly there are interesting ramp schemes where no participant is
important, thus this Corollary is not applicable. We sketch here a ramp
scheme which distributes infinitely many secrets, while every participant
has a finite share domain -- consequently no participant can be important.

In the scheme participants are indexed by the positive integers, and $X_s$
-- the set of secrets -- is also the set of positive integers. The dealer
chooses the secret $s\in X_s$ with probability $2^{-s}$. After choosing the
secret, she picks a \emph{threshold number $t>s$} with probability
$2^{-t+s}$. The participant with index $i\le t$ gets an integer from $[1,i]$
uniformly and independently distributed, participant with index $i>t$ gets
$s$ as the share.

The secret can be recovered by any infinite set of participants as the
eventual value of their shares, while any finite set is unqualified. It is
easy to see that this scheme is (measurable) ramp realizing all infinite
subsets of the positive integers, and has the required properties.


\section{Conclusion}\label{sec:conclusion}

In this paper we looked at the theoretical problems of infinite
probabilistic secret sharing schemes. It is quite natural to look at the
classical secret sharing schemes from a probabilistic point of view. While the
first few steps towards an abstract definition are easy, interesting and
unexpected phenomena appear quite early. The non-measurable scheme in
Section \ref{sec:non-measurable} was our first surprise. Without such a
``technical'' restriction as the measurability of the recovery function,
nothing can be said.

Some interesting infinite schemes in \cite{dibert-csirmaz:examples} do not
seem to fit into the security types defined in Section
\ref{subsec:sss-types}. Also, there are access structures which cannot be
realized by any measurable scheme which could be considered to be secure in
any sense shown by the following example.
\begin{quote}
Let $P$ be the lattice points in the positive quadrant; minimal qualified
sets are the ``horizontal'' lines.
\end{quote}
Suppose there are only two secrets (this can always be assumed). As the
first row is qualified, there are finitely many participants in the first
row who can determine the secret up to probability at least $0.9$.
Similarly, finitely many participants from the second row know the secret up
to probability $0.99$; finitely many from the third row with probability
$0.999$, etc. The union of these finite sets will know the secret with
probability $1$, thus this set will be qualified, while it intersects each
row in finitely many elements. (We actually showed that this structure is
not $G_\delta$ at the end of Section \ref{sec:sierpinski}.)

The nice, and surprisingly natural, characterization of ramp and weakly ramp
schemes in Section \ref{sec:ramp} hints that our definition is ``the'' right
one. As remarked earlier, no universally accepted definition exists for
weakly perfect, or ramp schemes. One flavor of definition uses entropies. If
$A$ is qualified, then the conditional entropy of the secret, given the
shares of $A$, is zero. If the shares of $B$ are independent of the secret,
then the conditional entropy equals the entropy of the secret. A scheme is
\emph{ramp}, if for unqualified subsets, this conditional entropy is never
zero. While this definition is widely applied in getting lower bounds on the
size of the shares in ramp schemes, it does not fit our definition. The
correct translation would be requiring the \emph{min-entropy} to be
positive: a classical scheme $\mathcal S$ is \emph{ramp} if for each value
the secret can take with positive probability, the conditional probability
of the same value for secret, given the value of the shares, is still
positive. In other words: in a ramp scheme unqualified subsets cannot
exclude any possible secret value (while the posterior probability that
the secret takes that value might be much smaller than the a priori
probability).

There are other interesting probabilistic schemes in
\cite{dibert-csirmaz:examples} which have weaker security guarantees than
weakly ramp schemes. In those schemes unqualified subsets can exclude large
subsets of the secret space, while still some uncertainty remains. A typical
example is where participant $i\in\mathbb N^+$ has a uniform random real
number from $[0,2^{-i}]$ as a share, and the secret is the sum of all
shares. If participant $i$ is missing, the rest can determine the secret up
to an interval of length $2^{-i}$, and within that interval the secret is
uniformly distributed. Is there any structure which can be realized by such
a scheme, but not by any ramp scheme? How can these scheme types be
captured by a definition similar to those in Definition
\ref{def:scheme-type}?

Finally we pose a question in another direction. Given an access structure
$\mathcal A$, is there an easy way to recognize whether it is $G_\delta$?
Given any collection of qualified and unqualified subsets, decide if there
is a $G_\delta$ structure separating them. As a concrete example: suppose
there is a collection of unqualified subsets of $P$ so that the union of any
two of them is qualified. Under what conditions is there a ramp scheme
realizing such a structure?

\vspace{6pt}

\section*{Acknowledgment}
\bgroup\small 

The author would like to thank the support and the uncountably many discussions while
developing the ideas in this paper of the members of the
Cryptography group at the R\'enyi Institute. Their input was indispensable in 
forming and correcting the ideas presented here. Special thanks go
to G\'abor Tardos who asked about measurability and constructed the
paradoxical example cited in this paper. Discussions with Bal\'azs Gyenis about 
Hilbert spaces clarified the characterization of ramp schemes.

The research reported in this paper was partially supported by GACR project
number 19-04579S, and by the Lend\"ulet program of the Hungarian
Academy of Sciences.

\egroup


\makecontacts


\begin{thebibliography}{99}
%
\bibitem{delarue}
Azéma, J.; Yor, M.; Meyer, P., de~la Rue; T.
\newblock Espaces de {L}ebesgue.
\newblock In {\em Séminaire de Probabilités XXVII}, volume 1557 of {\em
  Lecture Notes in Mathematics}, pages 15--21. Springer Berlin / Heidelberg,
  1993.
\newblock 10.1007/BFb0087958.
%
\bibitem{beimelsurvey}
Beimel, A.
\newblock Secret-sharing schemes: A survey.
\newblock In Yeow~Meng Chee, Zhenbo Guo, San Ling, Fengjing Shao, Yuansheng
  Tang, Huaxiong Wang, and Chaoping Xing, editors, {\em IWCC}, volume 6639 of
  {\em Lecture Notes in Computer Science}, pages 11--46. Springer, 2011.
%
\bibitem{blakley-swanson}
Blakley, G.R.; Swanson, L.
\newblock Infinite structures in information theory.
\newblock In {\em CRYPTO}, pages 39--50, 1982.
%
\bibitem{ChangPollard97}
Chang, J.T. ; D.~Pollard, D..
\newblock Conditioning as disintegration.
\newblock {\em Statistica Neerlandica}, 51(3):287--317, 1997.
%
\bibitem{ChorKush}
Chor, B.; Kushilevitz, E.
\newblock Secret sharing over infinite domain.
\newblock {\em Journal of Cryptology}, 6(2):97--86, 1993.
%
\bibitem{dibert-thesis}
Dibert, A.
\newblock Generalized secret sharing.
\newblock Master's thesis, Central European University, Budapest, Hungary,
  2011.
%
\bibitem{dibert-csirmaz:examples}
Dibert, A; Csirmaz, L.
\newblock Infinite secret sharing -- Exmples.
\newblock 2014.
\newblock {\em Journal of Mathematical Cryptology},  8(2);141--168, 2014.
\newblock 10.1515/jmc-2013-0005
%
\bibitem{fremlin} 
Fremlin, D.H.
\newblock {\em Measure Theory, Volume 2},
\newblock Torres Fremlin, Colchester, 2003. 563+12 pp
%
\bibitem{haezendonck}
Haezendonck, J.
\newblock Abstract {L}ebesgue--{R}okhlin spaces.
\newblock {\em Bulletin de la Societe Mathematique de Belgique}, 25:243--258,
  1973.
%
\bibitem{janson1997gaussian}
Janson, S.
\newblock {\em Gaussian Hilbert Spaces}.
\newblock Cambridge Tracts in Mathematics. Cambridge University Press, 1997.
%
\bibitem{kallenberg}
Kallenberg, O.
\newblock {\em Foundations of Modern Probability}.
\newblock Probability and Its Applications Series. Springer, 2010.
%
\bibitem{spanprogram}
Karchmer, M.; Wigderson, A.
\newblock On span programs.
\newblock In {\em Structure in Complexity Theory Conference}, pages 102--111,
  1993.
%
\bibitem{Makar:1980}
Makar, B.H..
\newblock Transfinite cryptography.
\newblock {\em Cryptologia}, 4(4):230--237, October 1980.
%
\bibitem{patarin:infinite}
Patarin, J.
\newblock Transfinite cryptography.
\newblock {\em IJUC}, 8(1):61--72, 2012.
\newblock also avaiable as \url{http://eprint.iacr.org/2010/001}.
%
\bibitem{Vaudenay:09}
Phan, R.; Vaudenay, S.
\newblock On the impossibility of strong encryption over $\aleph_0$.
\newblock In Yeow Chee, Chao Li, San Ling, Huaxiong Wang, and Chaoping Xing,
  editors, {\em Coding and Cryptology}, volume 5557 of {\em Lecture Notes in
  Computer Science}, pages 202--218. Springer Berlin / Heidelberg, 2009.
%
\bibitem{rokhlin}
Rokhlin, V.A.
\newblock On the fundamental ideas of measure theory.
\newblock {\em Translations (American Mathematical Society)}, 10:1–54, 1962.
%
\bibitem{tao-2011}
Tao, T.
\newblock {\em An introduction to measure theory}.
\newblock American Mathematical Soc.,
\newblock Graduate Studies in Mathematics, vol. 126
\newblock 2011; 206 pp.
%
\bibitem{Watson1990299}
Watson, S.
\newblock Power of the {S}ierpi{\'n}ski space.
\newblock {\em Topology and its Applications}, 35(2–3):299 -- 302, 1990.
%
\end{thebibliography}
\end{document}